\documentclass[11pt]{article}
\usepackage{amssymb}
\usepackage{amsmath}
\usepackage{amsthm}
\usepackage{latexsym}
\newtheorem{theo}{Theorem}[section]
\newtheorem{prop}[theo]{Proposition}
\newtheorem{cor}[theo]{Corollary}
\newtheorem{lemma}[theo]{Lemma}
\theoremstyle{definition}
\newtheorem{defi}[theo]{Definition}
\newtheorem{exa}[theo]{Example}
\newtheorem{rem}[theo]{Remark}
\numberwithin{equation}{section}

\newcommand{\N}{{\mathbb N}}
\newcommand{\F}{{\mathbb F}}

\newcommand{\Fqn}{\mathbb{F}_{q^n}}
\newcommand{\Fqns}{\mathbb{F}_{q^n}^*}

\newcommand{\Z}{{\mathbb Z}}

\newcommand{\cA}{{\mathcal A}}
\newcommand{\cC}{{\mathcal C}}
\newcommand{\cD}{{\mathcal D}}
\newcommand{\cG}{{\mathcal G}}

\newcommand{\cO}{{\mathcal O}}

\newcommand{\cL}{{\mathcal L}}

\newcommand{\cU}{{\mathcal U}}
\newcommand{\cV}{{\mathcal V}}
\newcommand{\cW}{{\mathcal W}}
\newcommand{\cX}{{\mathcal X}}
\newcommand{\cY}{{\mathcal Y}}

\newcommand{\im}{\mbox{\rm im}\,}

\renewcommand{\mod}{\mbox{\rm mod}\,}

\newcommand{\ideal}[1]{\mbox{$\langle{#1}\rangle$}}
\newcommand{\GL}{\mathrm{GL}}

\newcommand{\orb}{\textup{Orb}}
\newcommand{\orbb}{\textup{Orb}_{\beta}}
\newcommand{\ds}{\textup{d}_{\rm{S}}}

\newcommand{\stab}{\textup{Stab}}
\newcommand{\stabb}{\textup{Stab}_{\beta}}
\newcommand{\stabp}{\textup{Stab}^+}
\newcommand{\stabbp}{\textup{Stab}_{\beta}^+}
\newcommand{\mipo}{\textup{minpoly}}
\newcommand{\dsfrac}[2]{\displaystyle\frac{#1}{#2}}

\newcounter{alp}
\newcounter{ara}
\newcounter{rom}

\newenvironment{alphalist}{\begin{list}{(\alph{alp})\hfill}{\usecounter{alp}
     \topsep.5ex \labelwidth.6cm \leftmargin.6cm \labelsep0cm
     \rightmargin0cm \parsep0ex \itemsep0ex
     \partopsep1.6ex}}{\end{list}}
\newenvironment{arabiclist}{\begin{list}{(\arabic{ara})\hfill}{\usecounter{ara}
     \topsep0ex \labelwidth.6cm \leftmargin.6cm \labelsep0cm
     \rightmargin0cm \parsep0ex \itemsep0ex
     \partopsep1.6ex}}{\end{list}}

\begin{document}
\title{Cyclic Orbit Codes and Stabilizer Subfields}
\date{\today\\[5ex] \emph{Dedicated to the memory of Axel Kohnert (1962 -- 2013)}}
\author{Heide Gluesing-Luerssen, Katherine Morrison, Carolyn Troha\footnote{HGL was partially supported by the National
Science Foundation Grant DMS-1210061.
HGL and CT are with the Department of Mathematics, University of Kentucky, Lexington KY 40506-0027, USA;
\{heide.gl, carolyn.troha\}@uky.edu. KM is with the School of Mathematical Sciences, University of Northern Colorado, Greeley, CO 80639, USA;
Katherine.Morrison@unco.edu.}}

\maketitle

{\bf Abstract:}
Cyclic orbit codes are constant dimension subspace codes that arise as the orbit of a
cyclic subgroup of the general linear group acting on subspaces in the given ambient space.
With the aid of the largest subfield over which the given subspace is a vector space, the cardinality
of the orbit code can be determined, and estimates for its distance can be found.
This subfield is closely related to the stabilizer of the generating subspace.
Finally, with a linkage construction larger, and longer, constant dimension codes can be derived from cyclic
orbit codes without compromising the distance.

{\bf Keywords:} Random network coding, constant dimension subspace codes, cyclic orbit codes, group actions

{\bf MSC (2010):} 11T71, 94B60

\section{Introduction}\label{SS-Intro}
\setcounter{equation}{0}
Random network coding, introduced by Ahlswede et al.~in~\cite{ACLY00}, has proven to be a very effective tool for maximizing the information flow in
a non-coherent network with multiple sources and sinks.
The main feature of the network is that the nodes form random linear combinations of the incoming packets (vectors) and transmit the
resulting packets further  to their neighboring nodes.
As a consequence, the receiver nodes (sinks) of the network will obtain linear combinations of the packets that have been injected
into the network at its sources.

While this method is very effective in disseminating the information throughout the network, it is at the same time also highly
sensitive to error propagation. Due to the linear combinations of packets that are formed and transmitted
further, a single packet that has been corrupted (through noise, erasures, or injection of wrong packets by adversaries)
may contaminate all further packets.

In order to overcome this deficiency, K{\"o}tter and Kschischang~\cite{KoKsch08} developed an algebraic approach to
random network coding by considering messages as subspaces of some fixed vector space~$\F^n$.
This nicely captures the main feature of the network flow, namely the linear combinations of the packets.
In other words, codewords are now simply subspaces of~$\F^n$, and a code is a collection of such subspaces.
Transmitting information through the network is thus reformulated in terms of transmitting subspaces.
The relevant distance measure for this setting depends on the particular type of problem to be studied, but in essence
the distance between two subspaces amounts to the codimension of their intersection: the larger the intersection, the
smaller the distance of the subspaces.

Their ground-breaking paper~\cite{KoKsch08} initiated intensive research efforts on subspace codes, see
\cite{EKW10,EtSi09,EtVa11,GPB10,KSK09,KoKu08,RoTr13,SiKsch09,TMBR13} and the references therein.
In~\cite{SiKsch09}, Silva and Kschischang derive further details on the appropriate metrics for the various networks models, while
in~\cite{KSK09}, Khaleghi et al.\ determine various bounds for the cardinality and distance of subspace codes.
Some of these bounds are improved upon by Etzion and Vardy in~\cite{EtVa11}, and the authors also present further constructions of
subspace codes.
In~\cite{EtSi09}, Etzion and Silberstein present a construction of subspace codes with large distance and cardinality which
is based on rank-metric codes as introduced and studied earlier by Gabidulin in~\cite{Gab85}.
Decoding of such rank-metric codes is investigated by Gabidulin et al.\ in~\cite{GPB10}, and the results are further applied to
subspace codes for various combinations of errors and erasures.

The paper at hand is most closely related to the references~\cite{EKW10,EtVa11,KoKu08,RoTr13,TMBR13}.
All these papers study, or touch upon, cyclic orbit codes.
These are subspace codes that arise as an orbit of a subspace in~$\F^n$ under a cyclic subgroup
of~$\GL(n,\F)$.
If the group is irreducible, the code is called an \emph{irreducible cyclic orbit code}.
In this case, the group is conjugate to a subgroup of~$\Fqns$, and by considering the subspaces in the~$\F$-vector space~$\Fqn$,
one can utilize properties of the field extension~$\Fqn$.
In finite geometry, an element in $\GL(n,\F)$ of order~$q^n-1$ is called a Singer cycle.
They thus generate irreducible cyclic subgroups of $\GL(n,\F)$, and so their corresponding subspace codes are irreducible cyclic orbit codes.

In~\cite{KoKu08}, Kohnert and Kurz make use of the field extension~$\Fqn$ along with solving linear inequalities under a
Diophantine restriction in order to find subspace codes of constant dimension with large distance.
The interesting fact is that they impose a prescribed automorphism group on the putative solutions in order to reduce the system
of inequalities significantly, thus making the search feasible.
Using the cyclic group~$\F_{q^n}^*$, their method results in the union of cyclic orbit codes.
In \cite[Sec.~5]{KoKu08} the authors present their results for length $n=6,\ldots,14$, dimension $k=3$, and distance $\ds=4$.
Some of these codes improve upon the lower bounds for the cardinality that were known at that time, and
most of these codes are still the best ones known for given length, dimension, and distance.
In~\cite{EKW10}, Elsenhans et al. present a decoding algorithm for cyclic orbit codes of dimension~$3$.

In~\cite{EtVa11}, Etzion and Vardy introduce the notion of a \emph{cyclic subspace code}.
In our terminology this is a collection of subspaces in $\F^n=\F_{q^n}$ that is invariant under multiplication by a primitive
element~$\alpha$ of ~$\F_{q^n}$.
In other words, a cyclic subspace code is a union of cyclic orbit codes, potentially of different dimensions.
In~\cite[Sec.~III]{EtVa11} they present optimal cyclic codes of length~$8$ and~$9$ in the sense that there is no larger
cyclic code of the same length and distance.
Their example of length~$9$ improves upon the code of length~$9$, dimension~$3$, and distance~$4$
given by Kohnert and Kurz~\cite{KoKu08} because they were able to add a cyclic spread code
(hence cardinality $(2^9-1)/(2^3-1)=73$) to a collection of~$11$ cyclic orbit codes of cardinality~$2^9-1=511$.
This leads to an optimal cyclic constant dimension code of cardinality~$5694$.
This cardinality comes remarkably close to the bound $\cA_2(9,4,3)\leq 6205$, resulting from the Anticode bound, see~\cite[Thm.~1]{EtVa11}.

The above mentioned codes were all found by computer search based on
the deliberate choice of $\F_{q^n}^*$ as the automorphism group in order to make the search feasible.
Despite this restriction on the search, the codes found all come quite close to known bounds.
This indicates that cyclic orbit codes form a powerful class of constant dimension codes that needs to be investigated further.

Rosenthal and Trautmann~\cite{RoTr13} and Trautmann et al.~\cite{TMBR13} present an algebraic treatment of cyclic orbit codes
by combining the ideas of~\cite{KoKu08} with detailed methods from the theory of group actions.
They also extend their results to orbits under reducible cyclic groups.

We close this brief overview of the literature by mentioning that cyclic orbit codes also play a crucial role in the construction
of $q$-analogs of Steiner systems; for details we refer to~\cite{BEOVW13,EtVa11s} and the references therein.

In this paper we will study cyclic orbit codes generated by subspaces of~$\Fqn$ by
specifying the largest subfield of~$\Fqn$ over which the given subspace is a vector space.
This subfield, later called the \emph{best friend} of the code or subspace, is closely related to the stabilizer of the orbit.
Designing a subspace with a pre-specified best friend allows us to control cardinality and distance of the
orbit code.
In particular, we can give estimates on the distance in terms of the best friend.
Moreover, using the best friend we are able to compute the distance with the aid of multisets.
The computation improves upon earlier results in~\cite{KoKu08,RoTr13} by reducing the cardinality of the multiset.

Finally, we will present a construction that allows us to link cyclic orbit codes leading to longer and larger constant
dimension codes without compromising the distance.

\section{Preliminaries}\label{SS-Prelim}
We fix a finite field~$\F=\F_q$.
Recall that a \emph{subspace code of length~$n$} is simply a collection of subspaces in~$\F^n$.
The code is called a \emph{constant dimension code} if all subspaces have the same dimension.
The \emph{subspace distance} of a subspace code~$\cC$ is defined as
\[
    \ds(\cC):=\min\{\ds(\cV,\cW)\mid \cV,\,\cW\in\cC,\,\cV\neq\cW\},
\]
where the distance between two subspaces is
\begin{equation}\label{e-dist}
      \ds(\cV,\cW):=\dim\cV+\dim\cW-2\dim(\cV\cap\cW).
\end{equation}
This distance may be interpreted as the number of insertions and deletions of vectors that is needed in order to transform
a basis of~$\cV$ into a basis of~$\cW$.
It thus coincides with the corresponding graph distance (the length of the shortest path
from~$\cV$ to~$\cW$ in the graph with vertices being the subspaces of~$\F^n$ and where two subspaces are joined by an
edge if they differ by dimension one and the smaller one is contained in the larger one).
If $\dim\cV=\dim\cW=k$, then $\ds(\cV,\cW)=2(k-\dim(\cV\cap\cW))=2\big(\dim(\cV+\cW)-k\big)$.
As a consequence, if~$\cC$ is a constant dimension code of dimension~$k$ then
\begin{equation}\label{e-distC}
  \ds(\cC)\leq\min\{2k,\,2(n-k)\}.
\end{equation}

The \emph{dual} of a subspace code $\cC$ is defined as
\begin{equation}\label{e-Cdual}
   \cC^{\perp}:=\{\cU^\perp\mid \cU\in\cC\}.
\end{equation}
It is easy to see that $\ds(\cV^\perp,\cW^\perp)=\ds(\cV,\cW)$, and therefore $\ds(\cC)=\ds(\cC^\perp)$.

Two subspace codes~$\cC,\,\cC'$ of length~$n$ are called \emph{linearly isometric} if there exists an $\F$-linear isomorphism
$\psi:\F^n\longrightarrow\F^n$ such that $\cC'=\{\psi(\cU)\mid \cU\in\cC\}$.
This terminology stems from the fact that isomorphisms preserve dimensions of subspaces and thus preserve the distance between
any two subspaces.
Hence linearly isometric codes have the same subspace distance and even the same distance distribution, i.e., the list of all
distances between any two distinct subspaces in~$\cC$ coincides up to order with the corresponding list of~$\cC'$.
In \cite[Def.~9]{TMBR13} linear isometries are denoted as $\GL_n(\F)$-isometries.

Consider now the field extension~$\F_{q^n}$.
Since the $\F$-vector spaces $\F^n$ and $\F_{q^n}$ are isomorphic, we may consider subspace codes as collections of
subspaces in~$\F_{q^n}$.
In this paper we will focus specifically on subspace codes that are given as orbits under a particular group action.

In order to make this precise, we fix the following terminology.
An element~$\beta$ of~$\Fqn$ is called \emph{irreducible} if the minimal polynomial of~$\beta$ in~$\F[x]$, denoted by $\mipo(\beta,\F)$,
has degree~$n$.
Hence $\Fqn=\F[\beta]$.
As usual, we call~$\beta$ (and its minimal polynomial) \emph{primitive} if the multiplicative cyclic group generated by~$\beta$, denoted by
$\ideal{\beta}$, equals $\Fqns$.

We will study subspace codes that are derived from the natural action of the group $\ideal{\beta}$ on~$\Fqn$.
This action induces an action on the set of subspaces of~$\Fqn$, and thus gives rise to the following
type of constant dimension codes.
These codes were introduced in a slightly different form in~\cite{RoTr13,TMBR13};
we will comment on the relation to~\cite{RoTr13,TMBR13} after the definition.

\begin{defi}\label{D-OrbU}
Fix an irreducible element~$\beta$ of~$\Fqn$.
Let~$\cU$ be a $k$-dimensional subspace of the $\F$-vector space~$\Fqn$.
The \emph{cyclic orbit code} generated by~$\cU$ with respect to the group~$\ideal{\beta}\subseteq\Fqns$
is defined as the set
\begin{equation}\label{e-orbbU}
    \orbb(\cU):=\{\cU\beta^i\mid i=0,1,\ldots,|\beta|-1\}.
\end{equation}
The code $\orbb(\cU)$ is called \emph{primitive} if~$\beta$ is primitive.
\end{defi}

Obviously, a cyclic orbit code is a constant dimension code.

Let us briefly relate our approach to~\cite{EtVa11,RoTr13,TMBR13}.
Let $f=x^n+\sum_{i=0}^{n-1}f_i x^i\in\F[x]$ be the minimal polynomial of~$\beta$.
The~$\F$-vector spaces $\Fqn$ and $\F^n$ are isomorphic via the coordinate map with respect to the basis
$1,\,\beta,\,\ldots,\,\beta^{n-1}$.
In other words, we have the $\F$-isomorphism
\begin{equation}\label{e-FFn}
   \varphi:\Fqn\longrightarrow \F^n,\quad \sum_{i=0}^{n-1}a_i\beta^i\longmapsto (a_0,\ldots,a_{n-1}).
\end{equation}
Let~$M_f\in\GL_n(\F)$ be the companion matrix of~$f$, thus\footnote{Due to row vector notation, our companion
matrix is the transpose of the classical companion matrix.}
\begin{equation}\label{e-Mf}
   M_f=\begin{pmatrix} 0&1& & &  \\ & &1&  &  \\ & & &\ddots &  \\  & & & &1\\
      -f_0&-f_1&-f_2& \ldots& -f_{n-1}\end{pmatrix}.
\end{equation}
Since~$f$ is the minimal polynomial of~$\beta$, multiplication by~$\beta$ in~$\Fqn$ corresponds to multiplication by~$M_f$
in~$\F^n$ under the isomorphism~$\varphi$, i.e.,
\begin{equation}\label{e-betaM}
   \varphi(a\beta^i)=\varphi(a)M_f^i\text{ for all $a\in\Fqn$ and $i\in\N$.}
\end{equation}

Using the isomorphism~\eqref{e-FFn} of~$\Fqn$ with~$\F^n$, the orbit code~$\cC:=\orbb(\cU)$ takes the following form.
We can write~$\varphi(\cU)$ as $\varphi(\cU)=\im U:=\{xU\mid x\in\F^k\}$, i.e., the rowspace of~$U$, for a suitable matrix
$U\in\F^{k\times n}$ of rank~$k$.
Then $\varphi(\cU\beta^i)=\im(U M_f^i)$, where~$M_f$ is as in~\eqref{e-Mf}.
Thus, under the isomorphism~\eqref{e-FFn} the orbit code~$\orbb(\cU)$ simply becomes
\begin{equation}\label{e-UMf}
    \{\im(U M_f^i)\mid 0\leq i\leq|\beta|-1\}.
\end{equation}
In other words, the action of the cyclic group~$\ideal{\beta}\leq\Fqns$ on subspaces in~$\Fqn$ turns into the action of
the cyclic group $\ideal{M_f}\leq \GL_n(\F)$ on subspaces in~$\F^n$.

In~\cite{RoTr13,TMBR13} the authors introduce, more generally, orbit codes in~$\F^n$ with respect to a subgroup of~$\GL_n(\F)$.
These are subspace codes of the form $\{\im(U A)\mid A\in\cG\}$,
where~$U$ is any matrix of rank~$k$ in $\F^{k\times n}$ and $\cG$ a subgroup of~$\GL_n(\F)$.
The orbit code is called \emph{cyclic} if the group~$\cG$ is cyclic and \emph{irreducible} if~$\cG$ is irreducible, i.e.,
it does not have any nontrivial invariant subspaces in~$\F^n$.

It is easy to see that the cyclic group~$\ideal{M_f}\leq \GL_n(\F)$ is irreducible whenever~$f$ is an irreducible polynomial.
Hence the orbit codes in Definition~\ref{D-OrbU} are irreducible cyclic
orbit codes in the sense of~\cite{RoTr13,TMBR13}.

Furthermore, every irreducible matrix~$A\in\GL_n(\F)$ has an irreducible
characteristic polynomial, say~$g\in\F[x]$, and~$A$ is similar to the companion matrix~$M_g$; see also~\cite[Sec.~IV.A]{TMBR13}.
Thus $A=SM_gS^{-1}$ for some $S\in\GL_n(\F)$, and the irreducible cyclic subgroup~$\cG=\ideal{A}$ of~$\GL_n(\F)$
is conjugate to the cyclic subgroup~$\ideal{M_g}$.
As a consequence, the isomorphism of~$\F^n$ induced by~$S$ yields a linear isometry between
$\{\im(UA^i)\mid i\in\N\}$ and $\{\im(VM_g^i)\mid i\in\N\}$, where $V=U S$; see also~\cite[Thm.~9]{RoTr13}.

All of this shows that the study of irreducible cyclic orbit codes may be
restricted to orbit codes with respect to irreducible cyclic subgroups of~$\Fqns$, and these
are exactly the codes in Definition~\ref{D-OrbU}.
In this context, matrices of order~$q^n-1$ in the group~$\GL_n(\F)$ are also called \emph{Singer cycles}.
They thus correspond to the primitive elements of~$\Fqn$.

In~\cite[p.~1170]{EtVa11}, the authors introduce \emph{cyclic subspace codes} in~$\Fqn$.
These are codes that are invariant under multiplication by a primitive element, say~$\alpha$, of~$\Fqn$.
In other words, a cyclic subspace code is a union of primitive cyclic orbit codes, i.e.,
$\cC=\bigcup_{t=1}^T\hspace*{-2em}\raisebox{.4ex}{$\cdot$}\hspace*{2em} \orb_{\alpha}(\cU_t)$.
In~\cite{EtVa11} the authors do not require that~$\cC$ be a constant dimension code, hence $\cU_1,\ldots,\cU_T$ may have different
dimensions.

We close this section with the following simple fact.
\begin{rem}\label{R-Cperp}
The dual of an orbit code (in the sense of~\eqref{e-Cdual}) is an orbit code again.
Indeed, for any subspace $\cU\in\F^n$ and matrix $A\in\GL_n(\F)$ we have $(\cU A)^\perp=\cU^{\perp}(A^{\sf T})^{-1}$.
Moreover, $A^{\sf T}=SAS^{-1}$ for some $S\in\GL_n(\F)$, and therefore
$\orbb(\cU)^{\perp}$ is linearly isometric to $\orbb(\cU^{\perp})$; see also \cite[Thm.~18]{TMBR13}.
\end{rem}

As a consequence, we may and will restrict ourselves to cyclic orbit codes generated by a subspace~$\cU$ with $\dim\cU\leq n/2$.

\section{Stabilizer Subfield and Cardinality of Cyclic Orbit Codes}\label{SS-Basics}
Throughout this section we fix an irreducible element~$\beta$ of~$\Fqn$.

Consider a $k$-dimensional subspace~$\cU$ of~$\Fqn$ and its orbit code $\orbb(\cU)$.
In the following, we will mainly restrict ourselves to subspaces~$\cU$ that contain the identity~$1\in\Fqn$.
This will facilitate later considerations of the cardinality of the orbit code.
The restriction is not at all severe because if~$1\not\in\cU$ then for any nonzero element $u\in\cU$
the subspace $\tilde{\cU}:=\cU u^{-1}$ contains~$1$.
Since the $\F$-isomorphisms on~$\Fqn$ given by multiplication with nonzero constants commute,
$\tilde{\cU}\beta^i=\cU\beta^i u^{-1}$, and therefore multiplication by~$u^{-1}$ provides a linear isometry between
$\orbb(\cU)$ and~$\orbb(\tilde{\cU})$.
Note also that if $u^{-1}\in\ideal{\beta}$, e.g., if~$\beta$ is primitive, then $\orbb(\cU)$ and~$\orbb(\tilde{\cU})$ are equal.

Recall that the stabilizer of the subspace~$\cU$ under the action induced by~$\ideal{\beta}$ is defined as
\begin{equation}\label{e-stabU}
  \stabb(\cU):=\{\gamma\in\ideal{\beta}\mid \cU\gamma=\cU\}=\{\gamma\in\ideal{\beta}\mid \cU\gamma\subseteq\cU\}.
\end{equation}
This is clearly a subgroup of~$\ideal{\beta}$.
Let~$N\in\N$ be the minimal integer such that $\stabb(\cU) = \ideal{\beta^N}$.
Then $N$ is a divisor of $|\beta|$ and with the aid of the orbit-stabilizer theorem from group actions we have
\begin{equation}\label{e-N}
   |\stabb(\cU)|=\frac{|\beta|}{N},\quad
   \orbb(\cU)=\{\cU\beta^i\mid i=0,\ldots,N-1\},\quad
   |\orbb(\cU)|=N.
\end{equation}
We define the following subfield related to the stabilizer.
\begin{defi}\label{D-StabField}
Let $\stabbp(\cU)$ be the smallest subfield of~$\Fqn$ containing~$\F$ and the group $\stabb(\cU)$.
\end{defi}
It is clear that $\stabbp(\cU)$ is the field extension $\F[\beta^N]$, where $\ideal{\beta^N}=\stabb(\cU)$.
From this it follows immediately that~$\cU$ is a vector space over~$\stabbp(\cU)$.

We wish to briefly comment on the relation to representation theory.
\begin{rem}\label{R-ReprTh}
Let $G=\ideal{\gamma}=\stabb(\cU)$.
As we have already observed, $\stabbp(\cU)=\F[\gamma]$.
Via multiplication in~$\Fqn$, the group~$G$ has a natural $\F$-representation in $\GL(\cU)$, where $\GL(\cU)$
denotes the group of $\F$-automorphisms of~$\cU$.
As always for group representations, this naturally induces an $\F[G]$-module structure on~$\cU$.
In this case, this is the same as the canonical $\F[\gamma]$-vector space structure --
as we have already observed in the previous paragraph.
As a consequence, an $\F$-linear map~$f$ on~$\cU$ is $G$-linear (i.e., $f(gu)=gf(u)$ for all $g\in G$ and $u\in\cU$),
if and only if it is $\F[\gamma]$-linear.
In other words, the intertwining algebra of the above representation of~$G$ is given by
$\text{End}_{\F[\gamma]}(\cU)$.
Irreducibility of the representation (i.e.~$\cU$ does not have any non-trivial $G$-invariant subspaces)
is the same as saying that~$\cU$ does not have any non-trivial $\F[\gamma]$-subspaces.
Therefore the above representation is irreducible if and only if $\dim_{\F[\gamma]}(\cU)=1$.
\end{rem}

We fix the following notation.
For a primitive element~$\alpha$, thus $\ideal{\alpha}=\Fqn^*$, we drop the subscript~$\alpha$ and
simply write $\orb(\cU),\,\stab(\cU)$ and $\stabp(\cU)$.
The identities in~\eqref{e-stabU} and~\eqref{e-N} then read as
\begin{equation}\label{e-primnota}
  \stab(\cU)=\{\gamma\in\Fqns\mid \cU\gamma=\cU\},\ \orb(\cU)=\{\cU\alpha^i\mid i=0,\ldots,L-1\},
  \text{ where }L=\frac{q^n-1}{|\stab(\cU)|}.
\end{equation}
In this particular case, we have the following result.
\begin{lemma}\label{L-Stabplus}
Let~$\cU$ be a subspace of~$\Fqn$ such that $1\in\cU$.
Then $\stabp(\cU)=\stab(\cU)\cup\{0\}$ and $\stabp(\cU)$ is contained in~$\cU$.
Moreover,~$\cU$ is a vector space over $\stabp(\cU)$ with scalar multiplication being the multiplication
of the field~$\Fqn$.
\end{lemma}
\begin{proof}
We know that $\stab(\cU)=\{\gamma\in\Fqns\mid \cU\gamma=\cU\}$ is a subgroup of $\Fqns$ and
contains~$\F^*$.
Thus it remains to show that $\stab(\cU)\cup\{0\}$ is closed under addition.
Let $\gamma,\,\gamma' \in \stab(\cU)$, i.e., $\cU\gamma=\cU=\cU\gamma'$.
If $\gamma+\gamma' = 0$, then $\gamma +\gamma' \in \stab(\cU)\cup\{0\}$, and we are done.
Now let $\gamma+\gamma'\neq0$.
Then $\cU(\gamma+\gamma') \subseteq \cU\gamma+\cU\gamma'=\cU+\cU = \cU$, so $\gamma + \gamma' \in \stab(\cU)$.
All of this shows that $\stab(\cU)\cup\{0\}\subset\Fqn$ is closed under multiplication and addition,
making it a subfield, and in fact the smallest subfield containing $\stab(\cU)$.
Since $1\in\cU$, the stabilizer property shows immediately that $\stabp(\cU)$ is contained in~$\cU$ and
that~$\cU$ is a vector space over $\stabp(\cU)$ (see also Remark~\ref{R-ReprTh}).
\end{proof}

In the case where~$n$ is prime, the only proper subfield of~$\Fqn$ is~$\F$.
Thus, we have the following result.

\begin{cor}\label{C-nprime}
If~$n$ is prime, then $\stab(\cU)=\F^*$, and thus $|\orb(\cU)|=\frac{q^n-1}{q-1}$ for every proper subspace $\cU\subset\Fqn$.
\end{cor}

Let us now return to the general case where~$\beta$ is irreducible.
The containments $\stabb(\cU)\subseteq\stab(\cU)\subseteq\stabp(\cU)$ lead immediately to the following
situation for the non-primitive case.
\begin{cor}\label{C-Stabbetaplus}
For any irreducible~$\beta$, the subfield $\stabbp(\cU)$ is contained in~$\stabp(\cU)$.
Hence it is contained in~$\cU$  and~$\cU$ is a vector space over $\stabbp(\cU)$.
\end{cor}

The following example shows that the containment $\stabbp(\cU)\subseteq\stabp(\cU)$ may be strict.

\begin{exa}\label{E-F81}
Consider $\F=\F_3$ and $\Fqn=\F_{3^4}$.
Fix the primitive element~$\alpha$ with minimal polynomial $x^4+x+2$.
Consider $\beta:=\alpha^{16}$, which has order~$5$.
Then~$\beta$ is an irreducible element of~$\F_{3^4}$ because its minimal polynomial is found to be
$\mipo(\beta,\F)=x^4+x^3+x^2+x+1$.
Let~$\cU$ be the subfield~$\F_{3^2}$ (considered as a subspace of~$\F_{3^4}$).
Then clearly $\stabp(\cU)=\F_{3^2}$.
Moreover, since $1\in\cU$, any~$\gamma$ satisfying $\cU\gamma=\cU$ is already in~$\cU$.
But then the relative primeness of the orders of the groups~$\ideal{\beta}$ and $\F_{3^2}^*$ show that
$\stabb(\cU)=\{1\}$.
As a consequence, $\stabbp(\cU)=\F_3$.
Thus we see that $\stabbp(\cU)\subsetneq\stabp(\cU)$.
\end{exa}

Let us briefly consider an extreme case.
Since $\F_q^*$ is always contained in $\stab(\cU)$, we conclude that $|\stab(\cU)|\geq q-1$.
Using~\eqref{e-N} above, this leads to $|\orb(\cU)|=N\leq \frac{q^n-1}{q-1}$.
As a consequence, $N=q^n-1$ is possible only if $q=2$.
This also follows from Lemma~\ref{L-Stabplus}, because if $N=q^n-1$, then $\ideal{\alpha^N}=\{1\}$ and thus
$\stabp(\cU)=\{0,1\}=\F_2$.

\begin{prop}\label{P-Fqk}
Let~$\cU$ be a $k$-dimensional subspace of~$\Fqn$ such that $1\in\cU$.
Then
\[
    \frac{|\beta|}{\gcd(|\beta|,q^k-1)}\ \text{ divides }\ |\orbb(\cU)|.
\]
Assume now that~$k$ divides~$n$, and thus~$\F_{q^k}$ is a subfield of~$\Fqn$.
If~$\F_{q^k}^*\subseteq\ideal{\beta}$ then $\frac{|\beta|}{q^k-1}$ divides $|\orbb(\cU)|$.
Furthermore, $|\orbb(\cU)|= \frac{|\beta|}{q^k-1}$ if and only if $\cU=\F_{q^k}$.
Finally we have $\ds\big(\orbb(\F_{q^k})\big)=2k$.
\end{prop}

\begin{proof}
From Corollary~\ref{C-Stabbetaplus} we know that $\stabbp(\cU)=\F_{q^r}$ for some~$r$ and that~$\cU$
is a vector space over~$\F_{q^r}$.
Thus~$r$ divides~$k$ and so $q^r-1$ divides $q^k-1$.
Moreover, since $\stabb(\cU)$ is a subgroup of~$\F_{q^r}^*\cap\ideal{\beta}$ its order divides
$q^r-1$ as well as~$|\beta|$.
All of this shows that $|\stabb(\cU)|$ divides $\gcd(|\beta|,q^k-1)$, and now the first statement follows from
the identities in~\eqref{e-N}.

For the second part, note first that by assumption, $q^k-1$ divides $|\beta|$.
Thus the first statement is just a special case of the previous part.
For the rest, set $D:= \frac{|\beta|}{q^k-1}$.
\\
\noindent"$\Rightarrow$"
With the notation  as in~\eqref{e-N}, we have $D=N$.
Since $|\beta^N|=\frac{|\beta|}{\gcd(N,|\beta|)}=\frac{|\beta|}{N}=q^k-1$, the uniqueness of subgroups of a cyclic group gives us
$\ideal{\beta^N}=\F_{q^k}^*$.
Now the fact that $\ideal{\beta^N}=\stabb(\cU)$ with Corollary~\ref{C-Stabbetaplus} imply $\stabbp(\cU)=\F_{q^k}\subseteq \cU$.
Thus $\F_{q^k}= \cU$ due to dimension.
\\
\noindent"$\Leftarrow$"  Let $u\in\F_{q^k}^*$.
Then $(u\beta^D)^{q^k-1} = u^{q^k-1}\beta^{D\cdot (q^k-1)}=1\cdot 1=1$.
Since the nonzero elements of~$\F_{q^k}$ are exactly the roots of $x^{q^k-1}-1$ in~$\Fqn$,
we obtain $\F_{q^k}\beta^D=\F_{q^k}$.
Hence $|\orbb(\F_{q^k})|\leq D$.
Let $0\leq i< j< D$ and let $\gamma \in \F_{q^k}\beta^i \cap \F_{q^k}\beta^j$ with $\gamma \neq 0$.
Then $\gamma = \gamma_i\beta^i = \gamma_j\beta^j$, for some $\gamma_i, \gamma_j \in \F_{q^k}^*$.
But then $\beta^{j-i}=\gamma_i\gamma_j^{-1}\in\F_{q^k}^*$.
So $j-i\equiv 0$ mod $D$, which is impossible.
Thus $\F_{q^k}\beta^i \cap \F_{q^k}\beta^j = \{0\}$.
This shows that $|\orbb(\F_{q^k})| = D$, and using~\eqref{e-dist} we also see that $\ds\big(\orbb(\F_{q^k})\big)=2k$.
\end{proof}

We have the following special case of the previous result.
Recall the notation from~\eqref{e-primnota}.

\begin{cor}\label{C-Fqk}
Let~$\cU$ be a $k$-dimensional subspace of~$\Fqn$ such that $1\in\cU$. Then
\[
     |\orb(\cU)| = \frac{q^n-1}{q^k-1} \Longleftrightarrow \cU = \F_{q^k}.
\]
Furthermore, $\ds\big(\orb(\F_{q^k})\big)=2k$.
\end{cor}

\begin{rem}\label{R-Spread}
A subspace code~$\cC$ is called a \emph{spread} of~$\Fqn$ if $\cup_{\cV\in\cC}\cV=\Fqn$ and
$\cV\cap\cW=\{0\}$ for all distinct $\cV,\,\cW\in\cC$.
If only $\cV\cap\cW=\{0\}$ is true for all distinct $\cV,\,\cW\in\cC$, then~$\cC$ is called a \emph{partial spread}.
The previous result shows that $\orb(\F_{q^k})$ is a $k$-dimensional spread, and
$\orbb(\F_{q^k})$ is a partial spread for any irreducible element~$\beta$.
This result is also found in \cite[Thm.~11, Cor.~12]{RoTr13}.
\end{rem}

Let us return to Lemma~\ref{L-Stabplus}.
The following notion will arise repeatedly, so we introduce terminology for convenience.\footnote{
The best friend as in Definition~\ref{D-friend} is what is called the stabilizer subfield in the title of this paper.
However, since the terminology `friend' will also be needed frequently, we prefer `best friend' over the
more technical `stabilizer subfield'.}

\begin{defi}\label{D-friend}
Let~$\cU$ be a subspace of~$\Fqn$.
A subfield~$\F_{q^r}$ of $\Fqn$ is called a \emph{friend} of~$\cU$ if~$\cU$ is a
vector space over~$\F_{q^r}$ with scalar multiplication being the multiplication in the field $\F_{q^n}$.
The largest friend of~$\cU$  (with respect to cardinality) is called the \emph{best friend} of~$\cU$.
\end{defi}

Note that since~$\cU$ is a subspace of the $\F$-vector space~$\Fqn$, the field~$\F$ is a friend of~$\cU$,
and thus~$\cU$ also has a best friend.
\begin{rem}\label{R-BF}
For any subspace~$\cU$ of~$\Fqn$ and any friend~$\F_{q^r}$ of~$\cU$ we have $1\in\cU\Longleftrightarrow \F_{q^r}\subseteq\cU$.
\end{rem}

\begin{prop}\label{T-StabBF}
Let~$\cU$ be a $k$-dimensional subspace of~$\Fqn$ with $1\in\cU$.
Then the subfield $\stabp(\cU)$ is the best friend of~$\cU$.
Furthermore, any friend of~$\cU$ is contained in the best friend.
\end{prop}
\begin{proof}
We know from Lemma~\ref{L-Stabplus} that $\stabp(\cU)$ is a friend of~$\cU$.
Moreover, if~$\F_{q^l}$ is a friend of~$\cU$, then $\cU\gamma=\cU$ for all $\gamma\in\F_{q^l}^*$ by closure
of the scalar multiplication.
This implies $\F_{q^l}^*\subseteq\stab(\cU)$, hence $\F_{q^l}\subseteq\stabp(\cU)$.
\end{proof}

It is a consequence of the last result that all subspaces in~$\orb(\cU)$ have the same best friend, say~$\F_{q^r}$, and we
may therefore call~$\F_{q^r}$ the \emph{best friend of the orbit code}.

Example~\ref{E-F81} shows that we do not have an analogous characterization for $\stabbp(\cU)$, when~$\beta$ is not primitive.

The identities in~\eqref{e-N} now read as follows.
\begin{cor}\label{C-cardinality}
Let $\F_{q^r}$ be the best friend of $\cU$. Then
\[
    |\orb(\cU)|=\frac{q^n-1}{q^r-1}\ \text{ and }\ |\stab(\cU)|=q^r-1.
\]
\end{cor}

This result facilitates the design of orbit codes with a prescribed cardinality.

In the following we use $\dim_{\F_{q^l}}(\cU)$ for the dimension of a vector space~$\cU$ over the field~$\F_{q^l}$.
We also set  $\dim\cU:=\dim_{\F}\cU$.

\begin{exa}\label{E-EtVa1}
The following subspace~$\cU$ is taken from \cite[Ex.~1]{EtVa11}, where the distance and cardinality of the resulting
orbit code have been determined by straightforward testing and enumeration.
Consider $\F=\F_2$ and the field $\F_{2^6}$ with primitive element $\alpha$ having minimal polynomial $x^6+x+1\in\F[x]$.
Let $\cU:=\{0,\alpha^0,\alpha^1,\alpha^4,\alpha^6,\,\alpha^{16},\alpha^{24},\alpha^{33}\}$.
It is straightforward to check that this is a vector space over~$\F$ (generated by, for instance, $\{1,\alpha,\alpha^4\}$).
Using the isomorphism $\varphi:\sum_{i=0}^5 a_i\alpha^i\longmapsto (a_0,\ldots,a_5)$ between the vector spaces~$\F_{2^6}$ and
$\F_2^6$, see~\eqref{e-FFn}, the subspace~$\varphi(\cU)$ is given by
\[
   \varphi(\cU)=\im\begin{pmatrix}1&0&0&0&0&0\\0&1&0&0&0&0\\0&0&0&0&1&0\end{pmatrix}.
\]
Since $\dim\cU=3$ it is clear that~$\cU$ is not a vector space over the subfield $\F_{2^2}$.
Furthermore,~$\cU$ does not coincide with the subfield $\F_{2^3}$ because the latter contains the element~$\alpha^9$.
All of this shows that~$\F$ is the best friend of~$\cU$ and thus $|\orb(\cU)|=2^6-1=63$ by the last corollary.
\end{exa}

\begin{exa}\label{E-EtVa2}
Consider the field $\F_{2^{12}}$ and the primitive element~$\alpha$ with minimal polynomial
$x^{12}+x^7+x^6+x^5+x^3+x+1\in\F_2[x]$.
\begin{alphalist}
\item Since $\deg(\mipo(\alpha,\F_{2^2}))=6$, it is clear that
      $\cU:=\F_{2^2}+\alpha\F_{2^2}+\alpha^3\F_{2^2}$ is a direct sum, and thus~$\cU$
      is a $6$-dimensional subspace of the $\F_2$-vector
      space~$\F_{2^{12}}$.
      Obviously $\F_{2^2}$ is a friend of~$\cU$.
      Furthermore,~$\cU\neq\F_{2^6}$ because $\alpha\in\cU$, but~$\alpha\not\in\F_{2^6}$.
      Along with Proposition~\ref{T-StabBF}, all of this shows that~$\F_{2^2}$ is the best friend of~$\cU$ and thus
      $|\orb(\cU)|=(2^{12}-1)/(2^2-1)=1365$.
\item Similarly $\cW:=\F_{2^4}+\alpha\F_{2^2}$ is a $6$-dimensional subspace of~$\F_{2^{12}}$ with best friend $\F_{2^2}$.
      Thus $|\orb(\cW)|=1365$.
\end{alphalist}
\end{exa}

\section{The Subspace Distance of Cyclic Orbit Codes}
In the previous section we determined the cardinality of an orbit code in terms of the best friend.
Now we turn to the minimum distance of these codes, again making use of the best friend.

Throughout this section, we restrict ourselves to orbit codes with respect to the cyclic group~$\Fqn^*$.
Thus we fix a primitive element~$\alpha\in\Fqn^*$.
Moreover, let~$\cU$ be a $k$-dimensional subspace of~$\Fqn$.
We usually assume $k\leq n/2$ (see Remark~\ref{R-Cperp}), but will not make explicit use of this assumption.
Recall that the orbit code~$\orb(\cU)=\orb_{\alpha}(\cU)$  contains a subspace $\cU'$ such that $1\in\cU'$.
Therefore, we may assume without loss of generality that $1\in\cU$.
Finally, let~$\F_{q^r}$ be the best friend of~$\cU$, and define
\[
   t:=\frac{k}{r}=\dim_{\F_{q^r}}\cU.
\]
From Corollary~\ref{C-cardinality} we know that the cardinality of $\orb(\cU)$ is given by $N:=\frac{q^n-1}{q^r-1}$.
For the study of the subspace distance, we note that for orbit codes the minimum distance is given by
\[
  \ds(\orb(\cU))=\min\{\ds(\cU,\,\cU\alpha^j)\mid 1\leq j<|\orb(\cU)|\}.
\]
This follows directly from the identity $\ds(\cU\alpha^l,\cU\alpha^m)=\ds(\cU,\cU\alpha^{m-l})$.

\subsection{Bounds on the Subspace Distance via the Stabilizer}\label{SS-DistStab}
The following lemma serves to show the simple fact that the subspace distance is a
multiple of~$2r$.
This is so because the intersection of any two subspaces in $\orb(\cU)$ is a vector space over~$\F_{q^r}$.
\begin{lemma}\label{L-distance}
Define $s:=\max_{1\leq j<N}\dim_{\F_{q^r}}(\cU \cap \cU\alpha^j)$.
Then
\begin{equation}\label{e-dists}
   \ds(\orb(\cU)) = 2(k - sr)=2r(t-s).
\end{equation}
As a consequence,
\[
  2r\leq \ds\big(\orb(\cU)\big)\leq 2k.
\]
\end{lemma}

Of course, the upper bound $\ds\big(\orb(\cU)\big)\leq 2k$ is true for all constant dimension codes of dimension~$k$
(which by assumption is at most $n/2$) as we saw already in~\eqref{e-distC}.

\begin{proof}
Let $1\leq j< N$.
Clearly, $\cU\alpha^j$ and thus $\cU \cap \cU\alpha^j$ are vector spaces over~$\F_{q^r}$.
Let $s_j:=\dim_{\F_{q^r}}(\cU \cap \cU\alpha^j)$.
Since $1\leq j< N$, we know $\cU \neq \cU\alpha^j$, and therefore $0\leq s_j<t$.
Thus, $\ds(\cU,\cU\alpha^j) = 2(k - \dim(\cU\cap\cU \alpha^j))=2(k-s_jr)\geq2r(t-s)\geq2r$.
\end{proof}

Comparing the lower bound~$2r$ with Corollary~\ref{C-cardinality}, we observe the usual trade-off between the cardinality
of an orbit code and its (potential) distance:
the larger the best friend, the smaller the code, but the better the lower bound for the distance.

\begin{cor}\label{C-Spread}
For any orbit code $\orb(\cU)$ we have
\[
     \ds(\orb(\cU))=2k\Longleftrightarrow r=k\Longleftrightarrow\cU=\F_{q^k}.
\]
If any (hence all) of these properties are true, then~$\orb(\cU)$ is a spread code.
\end{cor}
\begin{proof}
Using the fact that $1\in\cU$ the second equivalence is obvious.
The implication ``$\Longleftarrow$'' of the first equivalence has been dealt with in Corollary~\ref{C-Fqk}.
As for ``$\Longrightarrow$'', note that Lemma~\ref{L-distance} implies that $\cU\alpha^j\cap\cU=\{0\}$ for all~$j$, hence
$\orb(\cU)$ is a partial spread.
Since $|\orb(\cU)|=(q^n-1)/(q^r-1)$, the union of all subspaces in the orbit results in $(q^k-1)(q^n-1)/(q^r-1)$
distinct nonzero points in~$\Fqn$.
Since $r\leq k$, this implies~$r=k$.
\end{proof}

The previous results have shown that the best distance for a $k$-dimensional primitive orbit code is~$2k$, in which
case the code is a spread.
On the other hand, Proposition~\ref{P-Fqk} tells us that these codes have the smallest cardinality among all
$k$-dimensional primitive orbit codes.

The next lemma shows that the worst distance, namely $\ds(\orb(\cU)) = 2r$, is attained whenever the defining
subspace~$\cU$ has a particularly regular form.

\begin{lemma}\label{L-dist-bad}
Suppose $\cU$ is of the form $\cU = \bigoplus_{i = 0}^{t-1} \alpha^{li} \F_{q^r}$ for some $1\leq l <\frac{q^n-1}{q^r-1}$, and
where~$\F_{q^r}$ is the best friend of $\cU$.
Then $\ds(\orb(\cU)) = 2r$.
\end{lemma}

\begin{proof}
Since $\alpha^l\cU = \bigoplus_{i = 1}^{t} \alpha^{li} \F_{q^r}$ we have
$\bigoplus_{i = 1}^{t-1} \alpha^{li} \F_{q^r}\subseteq \cU\cap\alpha^l\cU$.
Moreover, $l<|\orb(\cU)|$ yields
$\dim_{\F_{q^r}}(\cU\cap\alpha^l\cU)\leq t-1 = \dim_{\F_{q^r}}(\bigoplus_{i = 1}^{t-1} \alpha^{li} \F_{q^r})$.
So $\cU\cap\alpha^l\cU=\bigoplus_{i = 1}^{t-1} \alpha^{li} \F_{q^r}$, and
$\dim_{\F_{q^r}}(\cU\cap\alpha^l\cU) = t-1$, which is the maximum possible intersection between any two distinct
subspaces in the orbit code.
Hence in the notation of Lemma~\ref{L-distance} we have $s =t-1$, and $\ds(\orb(\cU))=2r$.
\end{proof}

Observe that in the previous lemma we added the requirement that~$\F_{q^r}$ be the best friend of~$\cU$ because this does
not follow from the form of~$\cU$.
Indeed, $\cU = \bigoplus_{i = 0}^{t-1} \alpha^{li} \F_{q^r}$ only implies that~$\F_{q^r}$ is a friend of~$\cU$, but it
may not be the best friend.
For instance, in~$\F_{2^6}$ with primitive element~$\alpha$ we have $\F_{2^2}=\F_2\oplus\alpha^{21}\F_2$, hence the
best friend is~$\F_{2^2}$.
The following result shows that this is essentially the only type of case where~$\F_{q^r}$ is not the best friend.

\begin{prop}\label{P-DirSumBF}
Let $\cU = \bigoplus_{i=0}^{t-1}\alpha^{il}\F_{q^r}$ for some~$l$, where $t>1$.
Then $\deg(\mipo(\alpha^l,\F_{q^r}))\geq t$.
Furthermore,
\begin{align*}
    \cU=\F_{q^{rt}}&\Longleftrightarrow  \deg(\mipo(\alpha^l,\F_{q^r}))=t\\
      &\Longleftrightarrow \alpha^l\cU=\cU\\
      &\Longleftrightarrow \text{$\F_{q^r}$ is not the best friend of~$\cU$.}
\end{align*}
In other words, $\F_{q^r}$ is the best friend of~$\cU$ if and only if~$\cU$ is not a field.
\end{prop}

\begin{proof}
First of all, the directness of the sum implies immediately that $\deg(\mipo(\alpha^l,\F_{q^r}))\geq t$.
As for the chain of equivalences we argue as follows.
\\
1) Assume $\cU=\F_{q^{rt}}$.
Then~$\cU$ is a field and the form of~$\cU$ shows that $\cU=\F_{q^r}[\alpha^l]$.
This implies $\deg(\mipo(\alpha^l,\F_{q^r}))=t$.
\\
2) $\deg(\mipo(\alpha^l,\F_{q^r}))=t$ yields $\dim_{\F_{q^r}}\F_{q^r}[\alpha^l]=t$, and
since~$\cU$ is contained in this field, we have $\cU=\F_{q^r}[\alpha^l]$.
This  implies $\alpha^l\cU=\cU$.
\\
3) If $\alpha^l\cU=\cU$, then $\alpha^l\in\stab(\cU)$ and hence $\alpha^l$ is contained in the best friend.
Since due to the directness of the sum,~$\alpha^l$ is not in~$\F_{q^r}$, we conclude that~$\F_{q^r}$ is not the best friend of~$\cU$.
\\
4) Assume that the best friend of~$\cU$ is~$\F_{q^{r'}}$ for some $r'>r$.
Set $\dim_{\F_{q^{r'}}}\cU=t'$.
Then $rt=k=r't'$.
We show that $\alpha^l\cU=\cU$.
Assume to the contrary that $\alpha^l\cU\neq\cU$.
Then $\dim_{\F_{q^{r'}}}(\cU\cap\alpha^l\cU)\leq t'-1$.
On the other hand $\bigoplus_{i=1}^{t-1}\alpha^{il}\F_{q^r}\subseteq(\cU\cap\alpha^l\cU)$.
Considering dimensions over~$\F=\F_q$ we obtain the inequality $r(t-1)\leq r'(t'-1)$, and
using $rt=r't'$ this yields $r\geq r'$, a contradiction.
Thus~$\alpha^l\cU=\cU$, and this implies that $\alpha^{lt}=\sum_{i=0}^{t-1}a_i\alpha^{li}$ for some $a_i\in\F_{q^r}$.
But this means that $\deg(\mipo(\alpha^l,\F_{q^r}))=t$ and $\cU=\F_{q^r}[\alpha^l]=\F_{q^{rt}}$.
\end{proof}

Of course, there are also subspaces that are not of the form in Lemma~\ref{L-dist-bad} and yet generate orbit codes with
distance as low as~$2r$.
\begin{exa}\label{E-non-optimal}
Consider~$\F_{2^{12}}$ with primitive element~$\alpha$ as in Example~\ref{E-EtVa2}.
Let $\cW=\F_{2^4}+\alpha\F_{2^2}$. In Example~\ref{E-EtVa2}(b) we saw that the best friend is~$\F_{2^2}$.
One can check that $\ds(\orb(\cW))=4=2r$.
\end{exa}

Let us now return to the case where the distance is large.
According to Lemma~\ref{L-distance} the best distance a non-spread orbit code may achieve is $2(k-r)$.
\begin{exa}\label{E-optimal}
\begin{alphalist}
\item The code in Example~\ref{E-EtVa1} is optimal among all non-spread orbit codes:
      in~\cite[p.~1170]{EtVa11} the distance has been found as~$4$, and this is $2(k-1)$.
\item Consider the code in Example~\ref{E-EtVa2}(a). In this case $k=6$ and $r=2$.
      One can verify that $\dim_{\F_{2^2}}(\cU\cap\cU\alpha^j) \leq 1$ for all
      $1\leq j< 1365=|\orb(\cU)|$.
      Hence Lemma~\ref{L-distance} yields $\ds(\orb(\cU))=2(k-r)=8$, which means the code
      is optimal among all non-spread orbit codes.
\end{alphalist}
\end{exa}

\begin{exa}\label{E-t=2}
Let $\dim_{\F_{q^r}}(\cU)=t=2$, hence $k=2r$. Then $2r=2(k-r)$, and thus
$\ds(\orb(\cU))=2(k-r)$ due to Lemma~\ref{L-distance}.
Thus any such code is optimal among all non-spread orbit codes with best friend~$\F_{q^r}$.
\end{exa}

Next we give a condition that leads to a distance less than $2(k-r)$.
Consider an intersection $\cV:=\cU\cap\cU\alpha^j$ for some~$j$.
Then~$\cV$ is an $\F_{q^r}$-subspace of~$\cU$, and thus~$\F_{q^r}$ is a friend of~$\cV$.
It is a consequence of Proposition~\ref{T-StabBF} that the best friend of $\cV$ is $\F_{q^{r'}}$
for some $r'$ such that $r\mid r'$.

\begin{prop}\label{P-dist-subspace}
Suppose there exists a subspace~$\cV$ of~$\cU$  with best friend~$\F_{q^{r'}}$ for some $r'>r$.
Then $\ds(\orb(\cU))\leq2(k-r')<2(k-r)$.
\end{prop}

\begin{proof}
Since $\F_{q^{r'}}$ is the best friend of $\cV$, Corollary~\ref{C-cardinality} yields
\[
   |\orb(\cV)|=\dsfrac{q^n-1}{q^{r'}-1}<\dsfrac{q^n-1}{q^r-1} = |\orb(\cU)|.
\]
So there exists some~$j$ such that $\cV\alpha^j = \cV$, while $\cU\alpha^j\neq\cU$.
Then $\cV\subset\cU\cap\cU\alpha^j$, so $\dim_{\F_{q^r}}(\cU\cap\cU\alpha^j)\geq \dim_{\F_{q^r}}(\cV) \geq\frac{r'}{r}$.
Hence $s:=\max_{1\leq j<N}\dim_{\F_{q^r}}(\cU\cap\cU\alpha^j)\geq\frac{r'}{r}$, and
Lemma~\ref{L-distance} implies 
$\ds(\orb(\cU))=2(k-sr)\leq2(k-r')<2(k-r)$.
\end{proof}

Example~\ref{E-non-optimal} illustrates this case.
The subspace~$\cW$ has best friend~$\F_{2^2}$, while also containing the field~$\F_{2^4}$.
Since that subspace is its own best friend, Lemma~\ref{L-distance} and Proposition~\ref{P-dist-subspace} guarantee
$2r=4\leq\ds(\orb(\cW))\leq 2(6-4)=4$, as we already saw.

We would like to stress that the condition in Proposition~\ref{P-dist-subspace} is not necessary for
the distance to be less than $2(k-r)$.
Lemma~\ref{L-dist-bad} provides examples of such codes.
For instance, the subspace $\cU$ of~$\F_{2^7}$ generated by $1,\,\alpha,\,\alpha^2$ (where~$\alpha$ is a primitive
element of~$\F_{2^7}$) has distance $\ds(\orb(\cU))=2r=2<2(k-r)$.
But since~$\F_2$ is the only subfield of~$\F_{2^7}$, every subspace of~$\cU$ has best friend~$\F_2$, and the assumption of
Proposition~\ref{P-dist-subspace} is not satisfied.

Unfortunately, we do not know any general construction of cyclic orbit codes with cardinality $(q^n-1)/(q^r-1)$ and distance $2(k-r)$, i.e.,
the best non-spread code case.
In~\cite[p.~7396]{TMBR13} it is conjectured that for any $n,k,q$ there exists a cyclic orbit code of cardinality
$(q^n-1)/(q-1)$ and distance $2(k-1)$.
In the same paper the conjecture is also verified for randomly chosen sets of
$(n,k,q)\in\{4,\ldots,100\}\times\{1,\ldots,10\}\times\{2,3\}$.

However, by exhausting all possible $4$-dimensional subspaces in~$\F_2^8$ via their row echelon form we could
verify that no cyclic orbit code exists with parameters $(n,k,r,q)=(8,4,1,2)$, hence with cardinality~$255$ and distance~$6$.
While there exists such a code for $(n,k,r,q)=(6,3,1,2)$ and distance~$4$, it remains open whether there is a cyclic
orbit code with parameters $(2k,k,1,q)$ and distance $2(k-1)$ for any~$k>4$.
The usual bounds, see e.g.~\cite{XiFu09}, do not rule out the existence of such codes.

\begin{exa}\label{E-bounds}
Let us consider cyclic orbit codes in $\F_{2^{12}}$ of dimension~$k=6$ and with best friend~$\F_2$.
Due to Corollary~\ref{C-cardinality}, such a code has cardinality~$2^{12}-1=4095$.
Because of the above discussion, we have doubts that there exists such a code with distance $2(k-1)=10$, but
we did not perform an exhaustive search.
The best code we could find with a random search has distance~$8$ and is generated by
$\cU=\F_2+\alpha\F_2+\alpha^4\F_2+\alpha^{10}\F_2+\alpha^{10}\beta\F_2+\alpha^8\beta^2\F_2$, where~$\alpha$ and~$\beta$ are primitive elements
of~$\F_{2^{12}}$ and $\F_{2^6}$, respectively.
\end{exa}

We close this section with the following positive observation.
Note that the codes found below will be used again in Example~\ref{E-PatchSecondBest} to build larger codes of the same quality.
\begin{exa}\label{E-k3SB}
It can be verified that for $q=2$ and all $n\in\{6,\ldots,20\}$, the cyclic orbit code $\orb(\cU)$ of
dimension~$k=3$ and cardinality $2^n-1$ with
\[
   \cU=\F_2+\alpha^2\F_2+\alpha^3\F_2\subseteq\F_{2^n}, \text{ where }\ideal{\alpha}=\F_{2^n}^*,
\]
has distance $4=2(k-1)$.
The same is true (maximal cardinality and distance~$4$) for $q=3,5,7$ and $n\in\{6,7,8\}$ and the analogous subspace~$\cU$.
We did not explore larger values of~$q$ and~$n$.
\end{exa}

\subsection{Computing the Subspace Distance via Multisets}
We now turn to a more explicit computation of the subspace distance of a cyclic orbit code.
The next result improves upon \cite[Thm.~15, Prop.~16]{RoTr13}, which in turn goes back to \cite[Lem.~1]{KoKu08}.
By taking the best friend into account, we will be able to work with a smaller multiset than in~\cite{RoTr13}, and we
do not have to distinguish between orbits of size $q^n-1$ (which can occur only if $q=2$) and those of smaller size.

As before let~$\cU$ have best friend~$\F_{q^r}$.
Lemma~\ref{L-Stabplus} yields
\begin{equation}\label{e-Ualpha}
    \stab(\cU)=\ideal{\alpha^N}=\F_{q^r}^*, \text{ where }N=\frac{q^n-1}{q^r-1}.
\end{equation}
Consider the group action $\Fqn \times \ideal{\alpha^N} \longrightarrow \Fqn$ given by $(v, \gamma)\mapsto v\gamma$.
For each $v\in\Fqns$ the orbit of~$v$ is
\[
   \cO(v):=\{v, v\alpha^N, v\alpha^{2N},\ldots, v\alpha^{(q^r-2)N}\},
\]
and $|\cO(v)|= |\ideal{\alpha^N}|=q^r-1$, since all elements of the orbit must be distinct.
Writing $v=\alpha^b$, we have
\[
    \cO(v) =\{ \alpha^{b}, \alpha^{b+N}, \ldots, \alpha^{b+N(q^r-2)}\}.
\]
By modular arithmetic there is exactly one element of this orbit whose exponent is strictly less than $N$.
Hence
\[
    \Fqns=\bigcup_{b=0}^{N-1}\hspace*{-1.1em}\raisebox{.5ex}{$\cdot$}\hspace*{1.1em}\cO(\alpha^b).
\]
Since $\cU$ is an~$\F_{q^r}$-vector space, the orbit $\cO(u)$ is in~$\cU$ for every $u\in\cU$.
This shows that
\begin{equation}\label{e-Uorbit}
     \cU\backslash\{0\} =\bigcup_{i=1}^S\hspace*{-.9em}\raisebox{.5ex}{$\cdot$}\hspace*{.9em}\mathcal{O}(\alpha^{b_i})\text{ for $S=\frac{q^k-1}{q^r-1}$
     and suitable non-negative integers }b_1,\ldots,b_S<N.
\end{equation}
One should note that $b_1,\ldots,b_S$ are uniquely determined by~$\cU$.
Moreover, if $\alpha^c\in \cU$ and $0\leq c<N$, then $c\in\{b_1,\ldots, b_s\}$.

For the following result, recall that a \emph{multiset} is collection of elements where each element is allowed to appear more than once.
We will denote multisets by double braces $\{\!\{ \ldots\}\!\}$.
The number of times an element appears in the multiset is called its \emph{multiplicity}.

\begin{theo}\label{T-dist}
Let~$\cU$ be as above and $b_1,\ldots,b_S$ be as in~\eqref{e-Uorbit}.
Define the multiset
\[
    \cD:=\{\!\{b_l-b_m~\mod N \mid 1\leq l,\,m\leq S,\,l\neq m\}\!\},
\]
and for $J\in\cD$ denote by $m(J)$ the multiplicity of~$J$ in~$\cD$.
Furthermore, set $M:= \max_{1\leq J<N}\{m(J)\}$.
If~$\cD=\emptyset$, we define $M:=0$.
Then $\dim(\cU\cap\cU\alpha^J) = \log_{q}(m(J)(q^r-1) +1)$ and
\[
   \ds(\orb(\cU)) = 2(k-L),\text{ where }L = \log_q(M(q^r-1)+1).
\]

\end{theo}

\begin{proof}
Let us first consider the case where~$\cD=\emptyset$.
This happens only if~$S=1$, hence $r=k$ and $\cU=\F_{q^k}$.
In this case $\ds(\orb(\cU))=2k$ as we know from Corollary~\ref{C-Spread}.

Let now $\cD\neq\emptyset$.
Fix $J\in\{1,\ldots,N-1\}$.
For all $l\in[S]:=\{1,\ldots,S\}$ we have $\alpha^{b_l+J}\in\cU\alpha^J$, and thus $\cO(\alpha^{b_l+J})\subset\cU\alpha^J$.
Hence $(\cU\alpha^J)\backslash\{0\} = \bigcup_{l\in[S]}\hspace*{-2.5em}\raisebox{.5ex}{$\cdot$}\hspace*{2.5em}\cO(\alpha^{b_l+J})$.
Since $\cU\cap\cU\alpha^J$ is an $\F_{q^r}$-vector space contained in~$\cU$, we have
\[
   (\cU\cap\cU\alpha^J)\backslash\{0\}=\bigcup_{l\in\cL_J}\hspace*{-1.2em}\raisebox{.5ex}{$\cdot$}\hspace*{1.2em}\cO(\alpha^{b_l}),
\]
where
\begin{align*}
  \cL_J&=\{l\in[S]\mid \cO(\alpha^{b_l})=\cO(\alpha^{b_m+J})\text{ for some }m\in[S]\}\\
       &=\{l\in[S]\mid \alpha^{b_l}=\alpha^{b_m+J}\alpha^{\lambda N}\text{ for some }m\in[S]\text{ and }\lambda \in\Z\}.
\end{align*}
Note that $\alpha^{b_l}=\alpha^{b_m+J}\alpha^{\lambda N}$ is equivalent to $b_l\equiv b_m+J+\lambda N~\mod (q^n-1)$.
Since~$N$ is a divisor of $q^n-1$, we conclude
\[
   \cL_J\subseteq\{l\in[S]\mid b_l-b_m\equiv J~\mod N\text{ for some }m\in[S]\}.
\]
By assumption there are $m(J)$ pairs $(b_l,b_m)$ so that $b_l-b_m\equiv J~\mod N$.
Thus, we obtain that $(\cU\cap\cU\alpha^J)\backslash\{0\}$ is the union of at most~$m(J)$ orbits.
This shows that $|\cU\cap\cU\alpha^J|\leq m(J)(q^r-1)+1$.

To show equality, note that there are $m(J)$ pairs $(b_l,b_m)$ such that $b_l-b_m\equiv J~\mod N$.
Pick such a pair $(b_l,b_m)$ and write $b_l=b_m+J+\lambda N$ for some~$\lambda\in\Z$.
Then $\cO(\alpha^{b_l})=\cO(\alpha^{b_m+J+\lambda N})=\cO(\alpha^{b_m+J})$, and so this orbit is in $\cU\cap\cU\alpha^J$.
This shows that there are $m(J)$ orbits in the intersection, and we conclude that
$|\cU\cap\cU\alpha^J|= m(J)(q^r-1)+1$.
Thus $\dim(\cU\cap\cU\alpha^J) = \log_{q}(m(J)(q^r-1) +1)$.

Finally, $\ds(\orb(\cU)) = 2(k-\max_{0<J<N}\{\dim(\cU\cap\cU\alpha^J)\})$, which leads to the desired result.
\end{proof}

\section{A Linkage Construction}\label{S-Link}
In this section we present a construction of how to use cyclic orbit codes to construct subspace codes with longer length.

In order to do so, it will be convenient to present subspaces as rowspaces of suitable matrices.
In other words, we will now consider subspaces in~$\F_q^n$ rather than $\F_{q^n}$.
Hence the orbit code $\orb_{\alpha}(\cU)$ becomes $\{\im U M^i\mid i=0,\ldots,|\orb_{\alpha}(\cU)|-1\}$
as in~\eqref{e-UMf}, where $M$ is the companion matrix of $\mipo(\alpha,\F_q)$, and $U\in\F^{k\times n}$ is
a matrix such that its rowspace $\im U$ is $\varphi(\cU)$ with the isomorphism~$\varphi$ in~\eqref{e-FFn}.

A constant dimension code in~$\F_q^n$ of dimension~$k$, cardinality~$M$, and subspace distance~$d$
will be called a $(n,M,d,k)_q$ code.
Recall also the notation $[N]:=\{1,\ldots,N\}$.

The following construction resembles the one given in the proof of~\cite[Thm.~11]{EtVa11}.
The latter, however, is tailored specifically to the use of a spread code as~$\cC_1$ and a specific choice of~$\cC_2$.
This allows a larger code~$\tilde{\cC}_3$ than our construction below.
In Theorem~\ref{T-PatchTwoCyclic} below we will, however, generalize \cite[Thm.~11]{EtVa11} by making appropriate choices.

\begin{theo}\label{T-PatchTwo}
For $i=1,2$ let~$\cC_i=\{\im U_{i,l}\mid l\in[N_i]\}$ be $(n_i,\,N_i,\,d_i,\,k)_q$ codes.
Thus~$U_{i,l}$ are matrices of rank~$k$ in $\F^{k\times n_i}$ for all $i,l$.
Define the subspace code $\cC_1\circledast\cC_2$ of length $n:=n_1+n_2$ as
$\cC_1\circledast\cC_2:=\tilde{\cC}_1\cup\tilde{\cC}_2\cup\tilde{\cC}_3$,
where
\begin{align*}
     &\tilde{\cC}_1=\{\im(U_{1,l},\,0_{k\times n_2})\mid l\in[N_1]\}, \\[.5ex]
     &\tilde{\cC}_2=\{\im(0_{k\times n_1},\, U_{2,l})\mid l\in[N_2]\}, \\[.5ex]
     &\tilde{\cC}_3=\{\im(U_{1,l},\,U_{2,m})\mid (l,m)\in[N_1]\times[N_2]\}.
\end{align*}
Then~$\cC_1\circledast\cC_2$ is a $(n,\,N,\,d,\,k)_q$ code, where $N=N_1+N_2+N_1N_2$ and
$d=\min\{d_1,\,d_2\}$.
\end{theo}

The reader should note that the code $\cC_1\circledast\cC_2$ depends not only on~$\cC_1$ and~$\cC_2$, but on the actual
choice of the matrices representing the subspaces.
Therefore the notation~$\cC_1\circledast\cC_2$ is not entirely correct, but this should not lead to any confusion.

\begin{proof}
The cardinality of~$\cC_1\circledast\cC_2$ is clear because the three sets~$\tilde{\cC}_i$ are pairwise disjoint.
Furthermore, it is obvious that $\ds(\tilde{\cC}_i)=\ds(\cC_i)$ for $i=1,2$.
Moreover, since each subspace in~$\tilde{\cC}_1$ intersects trivially with each subspace
in~$\tilde{\cC}_2$ or $\tilde{\cC}_3$, we conclude that
$\ds(\cW_1,\cW_2)=2k$ for all $\cW_1\in\tilde{\cC}_1$ and $\cW_2\in\tilde{\cC}_2\cup\tilde{\cC}_3$.
The same is true for the distance between subspaces in~$\tilde{\cC}_2$ and those in $\tilde{\cC}_1\cup\tilde{\cC}_3$.

It remains to consider the subspace distance between any two subspaces in~$\tilde{\cC}_3$.
Let $\cX=\im(U_{1,l},U_{2,m})$ and $\cY=\im(U_{1,l'},U_{2,m'})$ be two distinct subspaces in~$\tilde{\cC}_3$,
thus $l\neq l'$ or $m\neq m'$.
Let $(x_1,x_2)\in\cX\cap\cY$.
Then $x_1\in\im U_{1,l}\cap\im U_{1,l'}$ and $x_2\in\im U_{2,m}\cap\im U_{2,m'}$.
Moreover, $(x_1,x_2)= z(U_{1,l},U_{2,m})$ for some $z\in\F^k$, and this shows that
if $x_1=0$ or $x_2=0$, then $z=0$, and thus $(x_1,x_2)=0$.
This means that the projection of $\cX\cap\cY$ onto either component is injective.

All of this shows that
\[
  \dim(\cX\cap\cY)\leq\min\big\{\dim(\im U_{1,l}\cap\im U_{1,l'}),\,\dim(U_{2,m}\cap\im U_{2,m'})\big\}.
\]
Putting this together with the definition of the subspace distance in~\eqref{e-dist} implies $\ds(\cX,\,\cY)\geq\max\{d_1,\,d_2\}$, as desired.
\end{proof}

Now we may iterate the above construction.
Choosing another subspace code~$\cC_3$ of dimension~$k$, the code $(\cC_1\circledast\cC_2)\circledast\cC_3$ consists of the
rowspaces of matrices of the form
\[
  (U_{1,l},0,0),\ (0, U_{2,l},0),\ (U_{1,l},U_{2,m},0),\ (0,0,U_{3,l}),\ (U_{1,l},0,U_{3,m}),\ (0,U_{2,l},U_{3,m}),\
  (U_{1,l},U_{2,m},U_{3,r}).
\]
This also shows that~$\circledast$ is associative, and we simply write $\cC_1\circledast\cC_2\circledast\cC_3$.
Using $U_{i,0}$ for the $(k\times n_i)$-zero matrix, the result for the $t$-fold $\circledast$-product can be presented
in the following form.

\begin{theo}\label{T-PatchMany}
For $i=1,\ldots,t$ let $\cC_i=\{\im U_{i,l}\mid l\in[N_i]\}$ be $(n_i,\,N_i,\,d_i,\,k)_q$ codes.
Thus the matrices $U_{i,l}\in\F^{k\times n_i}$ have rank~$k$ for all $i=1,\ldots,t$ and $l\in[N_i]$.
Define $U_{i,0}:=0_{k\times n_i}$ for all~$i$.
Then the subspace code of length $n:=\sum_{i=1}^t n_i$
\[
  \cC_1\circledast\ldots\circledast\cC_t:=\big\{\im(U_{1,l_1},\,U_{2,l_2},\ldots,U_{t,l_t})\mid
       l_i\in\{0,\ldots,N_i\}\text{ for all }i,\, (l_1,\ldots,l_t)\neq(0,\ldots,0)\big\}
\]
is an $(n,N,d,k)_q$ code, where $N=\Big(\!\!\prod_{i=1}^t (N_i+1)\!\Big)-1$ and
$d=\min\{d_1,\ldots,d_t\}$.
\end{theo}
One may now ask what type of structure the resulting code has if the constituent codes are cyclic orbit codes.
A partial answer is provided with the following result.
It restricts to the case of primitive cyclic orbit codes over~$\F_2$ that each have~$\F_2$ as the best friend.

\begin{prop}\label{P-COCPatch}
Let $q=2$. For $i=1,\ldots,t$ let~$\alpha_i$ be a primitive element of~$\F_{2^{n_i}}$, and let
$\cC_i=\orb_{\alpha_i}(\cU)$ be a cyclic orbit code of dimension~$k$, length~$n_i$, and cardinality $2^{n_i}-1$.
If $\gcd(n_1,\ldots,n_t)=1$ then $\cC_1\circledast\ldots\circledast\cC_t$
is a union of cyclic orbit codes with respect to a fixed cyclic subgroup of $\GL_n(\F_2)$, where $n=\sum_{i=1}^t n_i$.
\end{prop}
Note that due to Theorem~\ref{T-PatchMany} the code $\cC:=\cC_1\circledast\ldots\circledast\cC_t$ has cardinality $2^n-1$.
But one should also observe that~$\cC$ is not necessarily cyclic with respect to the group $\F_{2^n}^*$.
As we will see in the proof, the cyclic subgroup of $\GL_n(\F_2)$ referred to in the theorem has order $\prod_{i=1}^t (2^{n_i}-1)$,
and this is less than $2^n-1$ if $t>1$.

\begin{proof}
Denote by $M_i\in\GL_{n_i}(\F)$ the companion matrix of the minimal polynomial of~$\alpha_i$ over~$\F_q$, see~\eqref{e-Mf}.
Then, as in~\eqref{e-UMf}, the codes~$\cC_i$ are given by
\[
    \cC_i=\{\im U_iM_i^l\mid l=0,\ldots,2^{n_i}-2\} \text{ for some suitable $U_i\in\F^{k\times n_i}$}.
\]
The code~$\cC:=\cC_1\circledast\ldots\circledast\cC_t$ has the form
\[
  \cC=\Big\{\im (V_1M_1^{l_1},\ldots,V_tM_t^{l_t})\mid V_i\in\{U_i,0_{k\times n_i}\},\,
      l_i\in\{0,\ldots,2^{n_i}-2\},\,(V_1,\ldots,V_t)\neq(0,\ldots,0)\Big\}.
\]
Using the notation $W_j\in\F^{k\times n},\,j=1,\ldots,2^t-1$, for the $2^t-1$ choices for $(V_1,\ldots,V_t)$, the code~$\cC$ is
the union of $2^t-1$ codes of the form
\[
   \cC^{(j)}=\Big\{\im \big(W_j\cdot\text{diag}(M_1^{l_1},\ldots,M_t^{l_t})\big)\,\Big|\, l_i\in\{0,\ldots,2^{n_i}-2\}\Big\}.
\]
Note that each $M_i$ generates a cyclic group of order~$2^{n_i}-1$.
Hence the direct product of these groups is a cyclic group if and
only if $\gcd(2^{n_1}-1,\ldots,2^{n_t}-1)=1$, which in turn is the case if and only if $\gcd(n_1,\ldots,n_t)=1$.
As a consequence, if $\gcd(n_1,\ldots,n_t)=1$, then each set $\cC^{(j)}$ is an orbit code with respect to this
cyclic group, and thus~$\cC$ is the union of cyclic orbit codes.
\end{proof}

We remark that the above does not provide a necessary condition for~$\cC$ to be a union of cyclic orbit codes
because it may happen that~$\cC$ is such a code without the lengths being relatively prime.
\begin{exa}\label{E-PatchSecondBest}
Let us use the binary subspace codes found in Example~\ref{E-k3SB}.
These are~$15$ codes,~$\cC_i$, of lengths $6,7,\ldots,20$.
Thus $\cC_1\circledast\ldots\circledast\cC_{15}$ has length $\sum_{i=6}^{20}i=195$, cardinality~$2^{195}-1$
and distance~$4$.
It is the union of $2^{15}-1$ cyclic orbit codes.
Note that the cardinality of this code equals the maximum cardinality of any binary cyclic orbit code of length~$195$.
\end{exa}

As we show next, applying the linkage construction to cyclic orbit codes results in a code whose cardinality
is upper bounded by the maximum cardinality of a cyclic orbit code of the same length.

\begin{prop}\label{P-PatchCard}
Let~$\cC_i,\,i=1,\ldots,t$, be as in Proposition~\ref{P-COCPatch}, and let
$n:=\sum_{i=1}^t n_i$.
Then $\cC:=\cC_1\circledast\ldots\circledast\cC_t$ satisfies $|\cC|\leq \frac{q^n-1}{q-1}$ with equality if and only if
$q=2$ and~$\F_2$ is the best friend of~$\cC_i$ for all~$i$.
\end{prop}
\begin{proof}
Let~$\F_{q^{r_i}}$ be the best friend of~$\cC_i$.
Then $|\cC_i|=\frac{q^{n_i}-1}{q^{r_i}-1}$.
Using Theorem~\ref{T-PatchMany}
we compute
\begin{equation}\label{e-Ccard1}
   |\cC|=\prod_{i=1}^t\Big(\frac{q^{n_i}-1}{q^{r_i}-1}+1\Big)-1\leq \prod_{i=1}^t\Big(\frac{q^{n_i}-1}{q-1}+1\Big)-1.
\end{equation}
We show that
\[
   \prod_{i=1}^t\Big(\frac{q^{n_i}-1}{q-1}+1\Big)\leq \frac{q^{\sum_{i=1}^t n_i}-1}{q-1}+1.
\]
In order to do so, we induct on~$t$.
It is clear that it suffices to consider the case~$t=2$.
Using that $\frac{q^{n_1}-1}{q-1}+1\leq q^{n_1}$, we compute
\begin{equation}\label{e-Ccard2}
  \Big(\frac{q^{n_1}-1}{q-1}+1\Big)\Big(\frac{q^{n_2}-1}{q-1}+1\Big)
  \leq q^{n_1}\frac{q^{n_2}-1}{q-1}+\frac{q^{n_1}-1}{q-1}+1
  =\frac{q^{n_1+n_2}-1}{q-1}+1,
\end{equation}
which is what we wanted.
This proves $|\cC|\leq \frac{q^n-1}{q-1}$.
Finally, we have equality in~\eqref{e-Ccard1} if and only if $r_i=1$ for all~$i$, and equality in~\eqref{e-Ccard2}
if and only if~$q=2$.
This concludes the proof.
\end{proof}

The cardinality of the linkage code in Theorem~\ref{T-PatchTwo} can be increased if one of the two
constituent codes is a subset of a cyclic orbit code.
In this case we can extend the code~$\tilde{\cC}_3$ by allowing all powers of the primitive element, regardless
of the stabilizer of the constituent orbit code.
As we will see in the proof, the linkage to the other constituent will guarantee that the distance is not compromised.
The construction in the proof of~\cite[Thm.~11]{EtVa11} by Etzion and Vardy may be regarded as the special case where
$N_2=1$ and~$\cC_1$ is a cyclic orbit spread code.

\begin{theo}\label{T-PatchTwoCyclic}
Let~$\cC_1=\{\im U_{1,l}\mid l\in[N_1]\}$ be a $(n_1,\,N_1,\,d_1,\,k)_q$ code,
thus $U_{1,l}\in\F^{k\times n_1}$ are matrices of rank~$k$.
Furthermore, let~$\alpha$ be a primitive element of~$\F_{q^{n_2}}$, and denote by~$M\in\GL_{n_2}(\F)$ the companion matrix of
the minimal polynomial of~$\alpha$.
Let $U_2\in\F^{k\times n_2}$ be a matrix of rank~$k$ and~$\cC_2$ be a subset of the cyclic orbit code $\orb(\im U_2)$ of length~$n_2$.
Hence there exists a set $\cL\subseteq\{0,1,\ldots,q^{n_2}-2\}$ such that $\cC_2=\{\im(U_2 M^l)\mid l\in\cL\}$ and
$\cC_2$ is a $(n_2,\,N_2,\,d_2,\,k)_q$ code where $N_2=|\cL|$.
Define the subspace code $\widetilde{\cC}$ of length~$n:=n_1+n_2$ as
$\widetilde{\cC}=\tilde{\cC}_1\cup\tilde{\cC}_2\cup\tilde{\cC}_3$,
where
\begin{align*}
     &\tilde{\cC}_1=\{\im(U_{1,l},\,0_{k\times n_2})\mid l\in[N_1]\}, \\[.5ex]
     &\tilde{\cC}_2=\{\im(0_{k\times n_1},\, U_2 M^l)\mid l\in\cL\}, \\[.5ex]
     &\tilde{\cC}_3=\{\im(U_{1,l},\, U_2 M^m)\mid (l,m)\in[N_1]\times\{0,\ldots,q^{n_2}-2\}\}.
\end{align*}
Then~$\widetilde{\cC}$ is a $(n,\,N,\,d,\,k)_q$ code, where $N=N_1+N_2+(q^{n_2}-1)N_1$ and
$d=\min\{d_1,\,d_2\}$.
\end{theo}
\begin{proof}
Comparing with Theorem~\ref{T-PatchTwo} and its proof, we see that the only case that remains to be considered is the distance between two
subspaces of the form $\cX=\im (U_{1,l},\,U_2M^m)$ and $\cY=\im(U_{1,l'},\,U_2M^{m'})$, where $l\neq l'$ or $m\neq m'$.
If $l\neq l'$, then $\im(U_{1,l})\neq\im(U_{1,l'})$, and as in the proof of Theorem~\ref{T-PatchTwo} we conclude
$\dim(\cX\cap\cY)\leq\dim\big(\im(U_{1,l})\cap\im(U_{1,l'})\big)$, and thus $\ds(\cX,\cY)\geq\max\{d_1,\,d_2\}$, as desired.

Let now $l=l'$. Then $m\neq m'$, but since $m,\,m'\in\{0,\ldots,q^{n_2}-2\}$, we may have that $\im U_2M^m=\im U_2M^{m'}$, and thus
we have to be more detailed.
Let $(x_1,x_2)\in\cX\cap\cY$, say $(x_1,x_2)=z(U_{1,l},\,U_2 M^m)=z'(U_{1,l},\,U_2 M^{m'})$.
Then the full row rank of~$U_{1,l}$ yields $z=z'$ and thus
$x_2=z(U_2 M^m)=z(U_2 M^{m'})$.
Using the isomorphism~$\varphi$ from \eqref{e-FFn}, this translates into
\[
  \varphi^{-1}(x_2)=c\alpha^m=c\alpha^{m'},\ \text{ where }c=\varphi^{-1}(z U_2).
\]
This is an identity in the field~$\Fqn$, and since $m,\,m'$ are less than the order of~$\alpha$, we have
$\alpha^m\neq\alpha^{m'}$, and so $c=0$.
Thus $z=0$ and $\cX\cap\cY=\{0\}$.
All of this shows that if~$l=l'$ then $\ds(\cX,\,\cY)=2k\geq\max\{d_1,d_2\}$, and this concludes the proof.
\end{proof}

We close the section with an example illustrating the construction.
\begin{exa}\label{E-PartSpread}
Let $k=3$.
In $\F_2^6$, choose the spread code $\cC_1=\orb(\F_{2^3})$.
This is a $(6,\,9,\,6,\,3)_2$ code, where the cardinality follows from Corollary~\ref{C-Fqk}.
In~$\F_2^7$ consider the cyclic orbit code $\cC'=\orb_{\alpha}(U)$, where~$\alpha$ is a primitive
element of~$\F_{2^7}$ and
\[
  U=\begin{pmatrix}1&0&0&0&0&0&0\\0&1&0&0&1&0&1\\0&0&1&1&0&1&0\end{pmatrix}.
\]
Let~$M\in\GL_7(\F_2)$ be the companion matrix of the minimal polynomial of~$\alpha$.
Then one can check that the subset~$\cC_2$ of~$\cC'$ given by
\[
   \cC_2=\big\{\im(U_2 M^j)\,\big|\, j\in\{0,2,5,10,20,23,57,72,75,91,95,109,113\}\big\}
\]
is a partial spread.
Thus~$\cC_2$ is a $(7,\,13,\,6,\,3)_2$ code.
Applying Theorem~\ref{T-PatchTwoCyclic} to~$\cC_1$ and~$\cC_2$ results in a
$(13,\,1165,\,6,\,3)_2$ code~$\widetilde{\cC}$, where the cardinality
stems from $9+13+9(2^7-1)=1165$.
This cardinality comes remarkably close to the upper bound $(2^{13}-2)/7-1=1169$ for
partial~$3$-spreads in~$\F_2^{13}$, see~\cite[Thm.~5]{EJSSS10}.
In fact, the difference of~$4$ is due to the fact that the partial spread~$\cC_2$ in~$\F_2^7$ is not maximal.
Again with~\cite{EJSSS10} it is known that one can find partial spreads in~$\F_2^7$ with~$17$ elements, as opposed to our
$13=|\cC_2|$.
While for $k=3$ and~$q=2$ maximal partial spreads are known for any~$n$ due to~\cite{EJSSS10}, we believe that
Theorem~\ref{T-PatchTwoCyclic} bears promising potential for larger~$k$.
\end{exa}

\section{Conclusion and Open Problems}
We have presented a detailed study of cyclic orbit codes based on the stabilizer subfield.
As has become clear, these codes have a rich structure which allows us to infer properties such as cardinality
and estimates on the distance.
While cyclic orbit codes themselves are in general not very large, taking unions of such codes
has resulted in constant dimension codes whose cardinalities come very close to known bounds.
Codes of this form have been found earlier in~\cite{EtVa11,KoKu08} via computer search; see the introduction of this paper
for further details.
Unfortunately, no systematic construction of unions of cyclic orbit codes or other orbit codes is known so far.
This and other observations in the previous sections lead to the following open problems.
\begin{arabiclist}
\item Find constructions of good cyclic subspace codes.
      In other words, find systematic ways to take unions of cyclic orbit codes without decreasing the distance.
\item Cyclic orbit codes with maximum distance, that is, $2k$, are spread codes (Corollary~\ref{C-Spread}) and thus well understood.
      The best distance a non-spread cyclic orbit code of dimension~$k$ can attain is thus $2(k-1)$, but a construction
      of such codes is not yet known.
      Therefore we formulate:
      For given~$n$ and $k\leq n/2$ construct cyclic orbit codes of length~$n$, dimension~$k$ and distance
      $2(k-1)$ or prove that no such code exists. See also Example~\ref{E-k3SB} and the paragraph right before Example~\ref{E-bounds}.
\item Make use of the algebraic structure of cyclic orbit codes in order to find efficient decoding algorithms.
      This has been initiated already in~\cite{EKW10,TMBR13}, but as discussed in the conclusion of~\cite{TMBR13} further research is needed.
\item Use other groups for constructing orbit codes. For instance, in~\cite{BEOVW13} the authors discovered
      non-trivial $q$-analogs of Steiner systems by testing unions of subspace orbits under the action of
      the normalizer group of $\F_{q^n}^*$ in $\GL(n,\F_q)$
      (thus under the combined action of~$\F_{q^n}^*$ and the Galois group $\text{Gal}(\F_{q^n}\mid\F_q)$).
\item In Section~\ref{S-Link} we have shown a linkage construction for general constant dimension codes and an improved version for
      cyclic orbit codes.
      We believe that this construction can be further enhanced by using suitable constituent codes.
\end{arabiclist}

\bibliographystyle{abbrv}

\end{document}